\documentclass[a4paper,12pt,reqno]{amsart}
\usepackage{amsmath,amsfonts,amssymb}
\usepackage[style=numeric]{biblatex}
\usepackage{amsthm}
\usepackage{verbatim}
\usepackage{enumerate}
\usepackage{xcolor}
\usepackage{float}
\usepackage{enumitem}
\usepackage{url}
\usepackage[colorlinks]{hyperref}
\usepackage[latin1]{inputenc}
\usepackage[english]{babel}
\usepackage{mathrsfs}
\usepackage[margin=2.5cm]{geometry}
\usepackage{bbm}
\usepackage{multirow}
\usepackage{mathrsfs}
\usepackage{caption}
\usepackage{float}
\restylefloat{table}
\usepackage{graphicx}

\def\Z{\mathbb Z}

\def\Ta{\mathcal T}
\def\ga{\mathcal G}

\def \F{\mathbb F}

\def\Fq{\mathbb{F}_q}

\DeclareMathOperator{\cyc}{Cyc}

\def\ord{\mathop{\rm ord}\nolimits}

\theoremstyle{plain}
\newtheorem{theorem}{Theorem}[section]
\newtheorem{lemma}[theorem]{Lemma}
\newtheorem{definition}[theorem]{Definition}

\newtheorem{proposition}[theorem]{Proposition}

\newtheorem{example}[theorem]{Example}

\author[J. A. Oliveira]{Jos\'e Alves Oliveira}
\author[F. E. Brochero Mart\'{\i}nez]{F. E. Brochero Mart\'{\i}nez}
\address{Departamento de Matem\'{a}tica,
	Universidade Federal de Minas Gerais,
	UFMG,
	Belo Horizonte MG (Brazil),
	30123-970}

\email{joseufmg@gmail.com}
\email{fbrocher@mat.ufmg.br }

\date{\today
}
\keywords{	Functional graph, polynomial maps, finite fields, finite abelian group}
\subjclass[2020]{37P25 (primary) and 05C20(secondary)}

\title{Dynamics of polynomial maps over finite fields}

\begin{document}

\begin{abstract}
		Let $\mathbb{F}_q$ be a finite field with $q$ elements and let $n$ be a positive integer. In this paper, we study the digraph associated to the map $x\mapsto x^n h(x^{\frac{q-1}{m}})$, where $h(x)\in\mathbb{F}_q[x].$ We completely determine the associated functional graph of maps that satisfy a certain condition of regularity. In particular, we provide the functional graphs associated to monomial maps.  As a consequence of our results, the number of connected components, length of the cycles and number of fixed points of these class of maps are provided.
\end{abstract}
	
\maketitle
	
\section{Introduction}

The iteration of polynomial maps over finite fields have attracted interest of many authors in the last few decades (for example see \cite{gomez2014irreducible,heath2019irreducible,oliveira2020iterations, reis2020factorization}). The interest for these problems has increased mainly because of their applications in cryptography, for example see \cite{johnson2001elliptic,wiener1998faster}. The iteration of polynomial map over a finite field yields a dynamical system, that can be related to its functional graph, which is formally defined as follows. Let $\Fq$ be a finite field with $q$ elements and let $f(x)\in\Fq[x]$. The functional graph associated to the pair $(f,\Fq)$ is the directed graph $\mathcal{G}\left(f / \Fq\right)$ with vertex set $V=\Fq$ and directed edges $A=\{(a,f(a)):a\in \Fq\}$.

While the iteration of polynomial maps has been widely studied, a complete characterization of their functional graphs has not been provided. Even though, many particular results in this direction are known. For example, the functional graphs and related questions are known for the following classes of polynomials:
\begin{enumerate}[label=\upshape(\roman*)]
	\item $f(x)=x^2$ over prime fields~\cite{rogers1996graph};
	\item $f(x)=x^n$ over prime fields~\cite{chou2004cycle};
	\item Chebyshev polynomials~\cite{gassert2014chebyshev};
	\item Linearized polynomials~\cite{panario2019functional}.
\end{enumerate}
Some other functions and problems concerning the dynamics of maps over finite structures has been of interest~\cite{peinado2001maximal,qureshi2015redei,qureshi2019dynamics,ugolini2018functional}. For a survey of the results in the literature, see~\cite{martins20198}. This paper's goal is to provide the functional graph associated to a class of polynomials in a general setting. For a polynomial $f(x)=x^n h(x^{\frac{q-1}{m}})$ with index $m$, we study the dynamics of the map $a\mapsto f(a)$ in order to present its functional graph $\mathcal{G}\left(f / \Fq\right)$.  Throughout the paper, we write
$$\mathcal{G}\left(f / \Fq\right)=\mathcal{G}^{ (0)}_{f / \Fq}\oplus\mathcal{G}^{(1)}_{f / \Fq},$$
where $\mathcal{G}^{ (0)}_{f / \Fq}$ denotes the connected component of $\mathcal{G}\left(f / \Fq\right)$ containing $0\in\Fq$. The paper's aim is to present explicitly the two functional graphs $\mathcal{G}^{ (0)}_{f / \Fq}$ and $\mathcal{G}^{(1)}_{f / \Fq}$ under a natural condition. The condition imposed along the paper guarantees that all the trees attached to cyclic points of $\mathcal{G}^{(1)}_{f / \Fq}$ are isomorphic. Our main results are essentially presented in two theorems. Theorem~\ref{item603} provides the component $\mathcal{G}^{ (0)}_{f / \Fq}$ that contains the vertex $0$ and Theorem~\ref{item601} provides the functional graph $\mathcal{G}^{(1)}_{f / \Fq}$ that contains all vertices that are not connected to the vertex $0$. For a polynomial $f(x)=x^n h(x^{\frac{q-1}{m}})$, the associated polynomial $\psi_f(x)=x^n h(x)^{\frac{q-1}{m}}$ will play an important role in the proof of our main results. In particular, the dynamics of the polynomial $f$ over $\Fq$ is established in terms of the dynamics of the map $\psi_f$ over the set $\mu_m\subset\Fq$ of $m$-th roots of the unity. For more details, see Section~\ref{item624}.

 This paper is organized as follows. In Section~\ref{item622} we present the terminology used along the paper and provide our main results. Section~\ref{item623} presents preliminary results that will be used throughout the paper. The proof of our main results is provide in Section~\ref{item624}.

\section{Terminology and main results}\label{item622}
In this section we fix the notation used in the paper, present our main results and provide some major comments. We use the same terminology as in \cite{qureshi2015redei,qureshi2019dynamics,qureshi2021functional}. Let $\alpha$ be a generator of the multiplicative group $\Fq^*$.   Along the paper, we will make an abuse of terminology by saying that the graphs are equal if they are isomorphic. By rooted tree, we mean a directed rooted tree, where all the edges point towards the root. Also, we use the letter $\Ta$ to denote a rooted tree. The tree with a single vertex is denoted by $\bullet$. We use $\cyc(k,\Ta)$ to denote a directed graph composed by a cycle of length $k$, where every node of the cycle is the root of a tree isomorphic to $\Ta$. The cycle $\cyc(k,\bullet)$ is also denoted by $\cyc(k)$. We use $\oplus$ to denote the disjoint union of graphs, and, for a graph $\ga$,  $k\times\ga$ denotes the graph $\oplus_{i=1}^k \ga $. If $\ga=\oplus_{i=1}^s \Ta_i$, where $\Ta_1,\dots,\Ta_s$ are rooted trees, then $\langle\ga\rangle$ represents the rooted tree whose children are roots of rooted trees isomorphic to $\Ta_1,\dots,\Ta_s$. We observe that the connected components of a functional graph related to the iteration of a function over a finite set consists of cycles where each vertex of the cyclic is the root of a rooted tree. 

In this paper, we study a class of polynomial maps whose functional associated graph has a regularity in the trees attached to each vertex in a cycle. In order to describe such trees, we present the well-known notion of elementary trees.

\begin{definition} For a non increasing sequence of positive integers $V=\left(v_{1}, v_{2}, \ldots, v_{D}\right)$, the rooted tree $\Ta_{V}$  is defined recursively as follows:
$$
\left\{\begin{array}{l}
	\Ta_{V}^{0}=\bullet,\\
	\Ta_{V}^{k}=\left\langle v_{k} \times \Ta_{V}^{k-1} \oplus \bigoplus_{i=1}^{k-1}\left(v_{i}-v_{i+1}\right) \times \Ta_{V}^{i-1}\right\rangle \quad \text { for } 1 \leq k\leq D-1,\\
	\Ta_{V}=\left\langle\left(v_{D}-1\right) \times \Ta_{V}^{D-1} \oplus \bigoplus_{i=1}^{D-1}\left(v_{i}-v_{i+1}\right) \times \Ta_{V}^{i-1}\right\rangle\text{ and }\\
	\Ta_{V}^{k}=\left\langle v_{k} \times \Ta_{V}^{k-1} \oplus \bigoplus_{i=1}^{k-1}\left(v_{i}-v_{i+1}\right) \times \Ta_{V}^{i-1}\right\rangle \quad \text { for }  k\geq D,
\end{array}\right.$$
where $v_i=v_D$ for all $i\ge D+1$.
\end{definition}

The graph $\Ta_{V}$ is called \textit{elementary tree}. Elementary trees play an important role in the study of functional graphs over finite fields, for example see~\cite{gassert2014chebyshev,panario2019functional,qureshi2015redei,qureshi2019dynamics,qureshi2021functional}. Along this text, the elementary trees will appear in our main statements. Indeed all the trees attached to vertices in cycles in the functional graphs arising from the maps we study are elementary trees.
For more details about this, see Lemmas~\ref{item610} and~\ref{item613}.

Throughout the paper, we use $\mu_m$ to denote the $m$-powers of $\Fq^*$ and $\Z_+$ to denote the set of positive integers. $\Z_d$ is the additive group of integers modulo $d$ and $\Z_d^*$ denotes the group of units of $\Z_d$. The order of an element $b\in\Z_d^*$ is denoted by $\ord_d(b)$. For a positive integer $d$, let
$$\mu(d)=\begin{cases}
	1,&\text{ if }d\text{ is square-free with an even number of prime factors;}\\
	-1,&\text{ if }d\text{ is square-free with an odd number of prime factors;}\\
	0,&\text{ if }d\text{ has a squared prime factor.}\\
\end{cases}$$
be the M{\"o}bius function.

In order to present our main results, we will follow the notation used in \cite{qureshi2021functional} denoting by $\gcd_n(v)$ the iterated $\gcd$ of $v$ relative to $n$, that is, $\operatorname{gcd}_{n}(v)=\left(v_{1}, \ldots, v_s\right)$, where
$$
v_{i}=\frac{\gcd(n^i,v)}{\gcd(n^{i-1},v)}\text{ for }i\ge 1
$$
and $s$ is the least positive integer such that $v_{s}=1$. This notion was introduced in \cite{qureshi2015redei}, where it was called $\nu$-series. An important property of $\operatorname{gcd}_{n}(v)=\left(v_{1}, \ldots, v_s\right)$ is that $v_1\cdots v_j=\gcd(n^j,v)$. This property will be used in the proof of our results.

Any polynomial $f(x)\in\Fq[x]$ satisfying $f(0)=0$ can be written uniquely as $f(x)=x^n h(x^{\frac{q-1}{m}})$, where $h(0)\neq 0$ and $m$ is minimal. The number $m$ is called the \textit{index} of the polynomial $f$. The index of polynomials play an important role in the study of polynomials over finite fields, for more details see~\cite{wang201915}. The index of polynomials over finite fields will be used in our main results. Along the paper, $f(x)=x^n h(x^{\frac{q-1}{m}})\in\Fq[x]$ is a polynomial with index $m$. We now present a notion that will be used to guarantee a certain regularity on the functional graph of $f(x)$.

\begin{definition}\label{item602}
	A polynomial $f(x)=x^n h(x^{\frac{q-1}{m}})\in\Fq[x]$ with index $m$ is said to be $m$-nice over $\Fq$ if the map $x\mapsto \psi_f(x)=x^n h(x)^{\frac{q-1}{m}}$ is an injective map from $\mu_m\backslash \psi_f^{(-1)}(0)$ to $\mu_m$.
\end{definition}
It is worth mentioning that $(\psi_f\circ x^{\frac{q-1}{m}})(a)=(x^{\frac{q-1}{m}}\circ f)(a)$ for all $a\in\Fq$. This fact is used along the proofs of our results. 
The map $\psi_f$ and  the fact that it commutes with $f$  play an important role in the study of permutation polynomials, for example see \cite{akbary2011constructing}.  In this paper, we present the dynamics of $f$ over $\Fq$ in terms of the dynamics of $\psi_d$ over $\mu_m$, that is usually a smaller set. In what follows, we present an example of polynomial that satisfy the notion of being nice. 

\begin{example}\label{item618}
	Let $g(x)=x^{15}h(x^{36})\in\F_{181}[x]$, where $h(x)=98x^4+68x^3+68x^2-6x-31$. Then $\psi_g(x)=x^{15}h(x)^{36}$. By straightforward computations, one can show that $\mu_5=\{1, 42, 59, 125, 135\}$ and $2$ is a primitive element of $\F_{181}$. Furthermore, 
	$$\psi_g(59)=42,\ \psi_g(42)=0,\ \psi_g(125)=125,\ \psi_g(1)=135\text{ and }\psi_g(135)=1.$$
	Therefore, $g(x)$ is $5$-nice over $\F_{181}$.
\end{example}

Throughout the paper, we let $\frac{q-1}{m}= \nu\omega$ where $\omega$ is the greatest divisor of $\tfrac{q-1}{m}$ that is relatively prime with $n$. Now we are able to present one of our main results. The following theorem provides the functional graph $\mathcal{G}^{(0)}_{f / \Fq}$ of $m$-nice polynomials.

\begin{theorem}\label{item603} Assume that $f(x)=x^n h(x^{\frac{q-1}{m}})$ is $m$-nice over $\Fq$ and let $d_j=\gcd(\nu,n^j)$. For each  $i=1,\dots,m+1$, let $r_i=|\{\xi\in\mu_m:\psi_{f}^{(i)}(\xi)=0\}|$. Then $\mathcal{G}^{(0)}_{f / \Fq}=\cyc(1,T)$, where
	$$T=\left\langle\bigoplus_{i=0}^{m-1} \left(\tfrac{(q-1)r_{i+1}}{md_{i}}-\tfrac{(q-1)r_{i+2}}{md_{i+1}}\right)\times\Ta_{\gcd_{n}(\nu)}^{i}\right\rangle.$$
\end{theorem}

We present now an example.

\begin{example} Let $g(x)\in\F_{181}[x]$ be defined as in Example~\ref{item618} and let notation be as in Theorem~\ref{item603}. Our goal is to apply Theorem~\ref{item603} for the polynomial $g(x)$. Since $n=15,m=5$ and $\tfrac{q-1}{m}=36$, we have that $\omega=4$,  $\nu=9$ and $\gcd_{n}(\nu)=(3,3,1)$. From Example~\ref{item618}, it follows that $r_1=r_2=1$ and $r_i=0$ for $i\ge 2$. Furthermore, $d_1=\gcd(9,15)=3$ and $d_i=\gcd(9,15^i)=9$ for $i\ge 2$. Therefore, Theorem~\ref{item603} states that $\mathcal{G}^{(0)}_{g / \F_{181}}=\cyc(1,T)$, where $$T=\left\langle 24 \times\Ta_{(3,3,1)}^{0}\oplus12\times \Ta_{(3,3,1)}^{1}\right\rangle.$$ Figure~\ref{item620} shows this functional graph.
	
	\begin{figure}[H]
		\centering
		\includegraphics[keepaspectratio,width=0.3\linewidth]{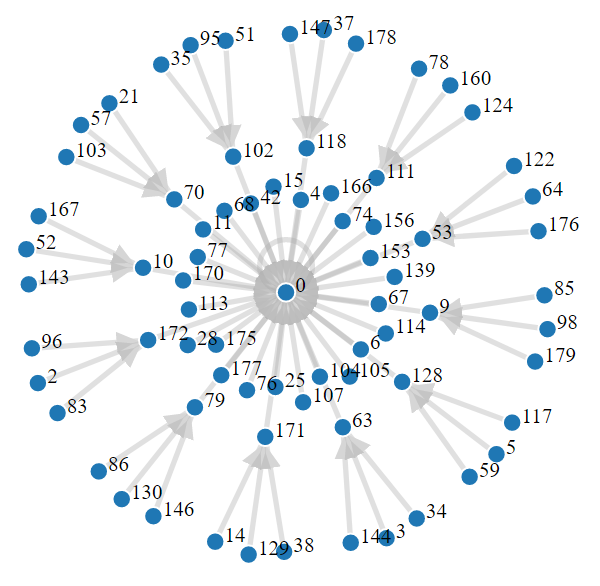}
		 \caption{\centering The conneted component of $\mathcal{G}(g/\F_{181})$ that contains the element $0\in\F_{181}$.}
		 \label{item620}
	\end{figure}
\end{example}
We now focus in the components of $\mathcal{G}(f/\Fq)$ that does not contains the element $0\in\Fq$. In order to present this graph, the following definition will be used.
\begin{definition} For a polynomial $f(x)\in\Fq[x]$ with index $m$, we define
	$$\psi_f^{-\infty}(0)=\{\gamma\in\mu_m:\psi_f^{(i)}(\gamma)=0\text{ for a positive integer }i\}.$$
\end{definition}

We note that if $f(x)=x^n h(x^{\frac{q-1}{m}})$ and $h(x)$ has no roots in $\mu_m$, then $\psi_f^{-\infty}$ is the empty set. Let $\omega'$ be the greatest divisor of $q-1$ that is relatively prime with $n$. In the following theorem, we determine the graph $\mathcal{G}^{(1)}_{f / \Fq}$ under the hypothesis that $f(x)$ is $m$-nice.

\begin{theorem}\label{item601}  Assume that $f(x)=x^n h(x^{\frac{q-1}{m}})$ is $m$-nice over $\Fq$ and let $S_1,\dots,S_t\subset\mu_m$ be sets such that $\mu_m\backslash \psi_f^{-\infty}(0)=S_1\cup\dots\cup S_t$ and $\mathcal{G}\left(\psi_f / S_i \right)=\cyc(k_i)$. For each $i=1,\dots,t$, let $\xi_i\in S_i$ and $\ell_i,r_i$ be integers such that $\xi_i=\alpha^{\frac{q-1}{m} r_i}$ and $\alpha^{\ell_i}=\prod_{j=0}^{k_i-1}h\big(\psi_f^{(j)}(\xi_i)\big)^{n^{k_i-j-1}}$. Then
	$$\mathcal{G}^{(1)}_{f / \Fq}=\bigoplus_{\substack{ i=1,\dots,t \\ u\mid \ord_{\scalebox{0.5}{$w'(n^{k_i}-1)$}}(n^{\scalebox{0.4}{$k_i$}})\\}}\left(\frac{\sum_{d\mid u}\mu(\tfrac{u}{d})\tau_i(d)}{u} \times \cyc\left(k_i u, \Ta_{\gcd_{n}(\nu)}\right)\right), $$
	where {\small $$\tau_i(d):=\begin{cases}
			\gcd\big(\tfrac{q-1}{m},n^{dk_i}-1\big),&\text{ if }\gcd\big(q-1,(n^{dk_i}-1)m\big)\mid \big(l_i\big(\tfrac{n^{dk_i}-1}{n^{k_i}-1}\big)+r_i(n^{dk_i}-1)\big);\\
			0,&\text{ otherwise.}\\
		\end{cases}$$ }
\end{theorem}

We present now an example.

\begin{example}  Let $g(x)\in\F_{181}[x]$ be defined as in Example~\ref{item618} and let notation be as in Theorem~\ref{item601}. Our goal is to apply Theorem~\ref{item601} for the polynomial $g(x)$.  From Example~\ref{item618}, we can choose $S_1=\{125\}$ and $S_2=\{1,135\}$, so that $k_1=1$ and $k_2=2$. Furthermore, $w'=4$, $\ord_{\scalebox{0.7}{$4(15-1)$}}(15)=2$, $\ord_{\scalebox{0.7}{$4(15^2-1)$}}(15^{\scalebox{0.7}{$2$}})=4$, $\ell_1=18$, $\ell_2=75$, $r_1=3$ and $r_2=0$, which implies that $$\tau_1(1)=2,\ \tau_1(2)=4,\ \tau_2(1)=0,\ \tau_2(2)=0\text{ and }\tau_2(4)=4.$$ Therefore, Theorem~\ref{item601} states that 
	$$\mathcal{G}^{(1)}_{g / \F_{181}}=2\times \cyc(1,\Ta_{(3,3,1)})\oplus \cyc(2,\Ta_{(3,3,1)})\oplus\cyc(8,\Ta_{(3,3,1)}).$$
	Figure~\ref{item621} shows this functional graph.
	\begin{figure}[H]
		\centering
		\includegraphics[keepaspectratio,width=0.5\linewidth]{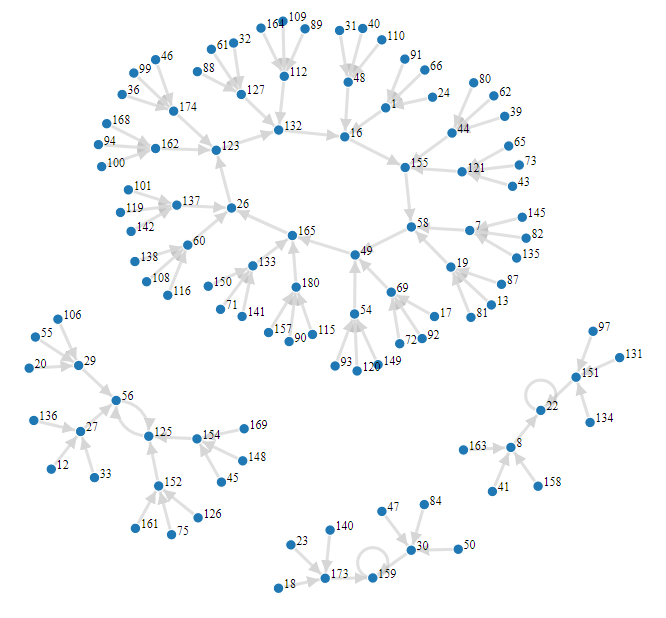}
		\caption{\centering The conneted components of $\mathcal{G}(g/\F_{181})$ that does not contain the element $0\in\F_{181}$.}
		\label{item621}
	\end{figure}
\end{example}
We observe that in the case where $m=1$ and $h(x)=ax$, then $f(x)=a x^n$ for all $x\in\Fq$. Furthermore, $f(x)=ax^n$ is $1$-nice, which implies that Theorems~\ref{item601} holds. In this case, Theorems~\ref{item601} reads
	$$\mathcal{G}^{(1)}_{f / \Fq}=\bigoplus_{u\mid \ord_{\scalebox{0.5}{$w'(n-1)$}}(n)}\left(\frac{\sum_{d\mid u}\mu(\tfrac{u}{d})\tau(d)}{u} \times \cyc\left(u, \Ta_{\gcd_{n}(\nu)}\right)\right),$$
	where $\alpha^\ell=a$ and
	{\small $$\tau(d):=\begin{cases}
			\gcd\big(q-1,n^{d}-1\big),&\text{ if }\gcd\big(q-1,(n^{d}-1)\big)\mid l\big(\tfrac{n^{d}-1}{n-1}\big);\\
			0,&\text{ otherwise.}\\
		\end{cases}$$ }
This expression generalizes some results obtained in \cite{chou2004cycle,qureshi2015redei,qureshi2021functional}. 

O note that, in particular, Theorems~\ref{item603} and \ref{item601} gives the number of connected components, the length of the cycles and the number of fixed points of $m$-nice polynomials. In the case where this condition is not satisfied, the functional graph of the polynomial is more chaotic, what makes it difficult to use the same approach used here. In the following example we present a polynomial that is not nice and its associated functional graph.

\begin{example}
	Let $g(x)=x^{6}h(x^{24})\in\F_{97}[x]$, where $h(x)=x-1$. Then $\psi_g(x)=x^{6}(x-1)^{24}$. One can show that $\mu_5=\{1,22,75,96\}$. Furthermore, 
	$$\psi_g(22)=\psi_g(75)=22,\ \psi_g(96)=96\text{ and }\psi_g(1)=0.$$
	Therefore, $g(x)$ is not $4$-nice over $\F_{97}$. In this case, the trees attached to cyclic vertices of $\mathcal{G}(g / \F_{97})$ has no regularity. Figure~\ref{item656} shows this functional graph.
	\begin{figure}[H]
		\centering
		\includegraphics[keepaspectratio,width=0.7\linewidth]{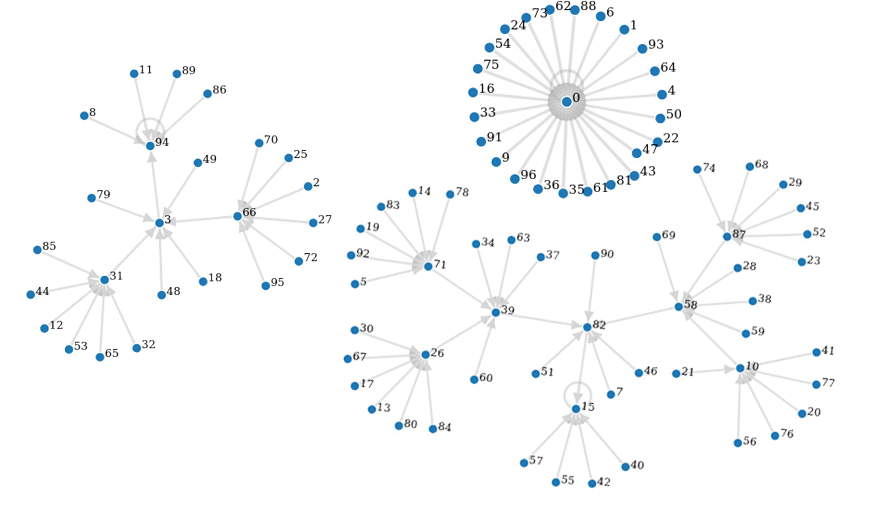}
		\caption[The functional graph $\mathcal{G}(g / \F_{97})$]{\centering The functional graph $\mathcal{G}(g / \F_{97})$}
		\label{item656}
	\end{figure}
\end{example}

\section{Preparation}\label{item623}

In this section, we provide preliminary notations and results that will be important in the proof of our main results. Let $b$ and $c$ be vertices in a directed graph. We say that a vertex $c$ is $k$-distant from a vertex $b$ if $c$ is at a distance of $k$ from $b$.  If a vertex $b$ is reachable from $c$, then $c$ is a predecessor of $b$ and $b$ is a successor of $c$. The following notions and results are the tools we need to prove Theorems~\ref{item603} and~\ref{item601}.

\begin{definition}
	Let $\ga$ be a directed graph, $V=\left(v_{1}, v_{2}, \ldots, v_{D}\right)$ be a non increasing sequence of positive integers and let $v_i=v_D$ for all $i\ge D$. We say that $\ga$ is $V$-regular if for each positive integer $k$, the number of $k$-distant predecessors of a vertex of $\ga$ is either $0$ or $v_1\cdots v_k$.
\end{definition}

\begin{lemma}\label{item610}
	Let $V=\left(v_{1}, v_{2}, \ldots, v_{D}\right)$ be a non increasing sequence of positive integers. If $\Ta$ is a $V$-regular rooted tree with depth $k$, then $\Ta$ is isomorphic to $\Ta_V^k$. 
\end{lemma}

\begin{proof}
	For $i\ge D+1$, let $v_i=v_D$. We proceed by induction on $k$. The case $k=0$ follows directly. Suppose that the result is true for an integer $k\ge 1$ and $\Ta$ is a $V$-regular rooted tree with depth $k+1$. Let $T$ be the rooted tree obtained from $\Ta$ by deleting the vertices with depth $k+1$. By induction hypothesis, $T$ is isomorphic to $\Ta_V^k$.
	
	 Let $z_1,\dots,z_{v_{k+1}}$ be the vertices of $\Ta$ with depth $1$ that have at least one descendant with depth $k+1$ and let $z_1',\dots,z_{v_{k+1}}'$ be the equivalent vertices in $T$. For $i=1,\dots,v_{k+1}$, let $\Ta_{z_i}$ be the rooted tree obtained from $\Ta$ containing all descendants of $z_i$ and let $T_{z_i}$ be the rooted tree obtained from $T$ containing all descendants of $z_i'$. Since $\Ta$ is $V$-regular, each $\Ta_{z_i}$ is isomorphic to $\Ta_V^k$. Therefore, $\Ta$ can be recovered from $T$ by replacing each $T_{z_i}$ by $\Ta_V^k$. The tree obtained from these steps is isomorphic to the rooted tree $\Ta_V^{k+1}$, which completes the proof of our assertion.
\end{proof}

\begin{definition}
	For a vertex $b$ of a directed graph $\mathcal G$, the graph $R_b(\ga)$ is the subgraph of $\ga$ containing all predecessors of $b$ (including $b$). 
\end{definition}

\begin{lemma}\label{item613}
	Let $V=\left(v_{1}, v_{2}, \ldots, v_{D}\right)$ be a non increasing sequence of positive integers such that $v_D=1$, let $\ga$ be a $V$-regular directed graph with depth $\ge D$, $b$ be a vertex of $\ga$ and $c$ a child of $b$. Let $G$ be the graph obtained from $R_b(\ga)$ by removing $c$ and its predecessors. If $G$ is a tree, then $G$ is isomorphic to $\Ta_{V}$.
\end{lemma}

\begin{proof}
	It follows similarly to the proof of Lemma~\ref{item610}.
\end{proof}

If $f(x)=x^n h(x^{\frac{q-1}{m}})$ is a polynomial with index $m$ and $k$ is a positive integer, then one can readily prove that
\begin{equation}\label{item606}
	f^{(k)}(x)=x^{n^k} \prod_{i=0}^{k-1} h\left(\psi_f^{(i)}\left(x^{\frac{q-1}{m}}\right)\right)^{n^{k-i-1}}.
\end{equation}
This formula will be important in the proofs of the main results.

\begin{lemma}\label{item604}
	Let $k$ be a positive integer and $b\in\Fq^*$. Assume that $f(x)=x^n h(x^{\frac{q-1}{m}})$ is $m$-nice. If $x\in\Fq$ is a solution of the equation $f^{(k)}(x)=b$, then $x^{\frac{q-1}{m}}=\psi_f^{(-k)}\big(b^{\frac{q-1}{m}}\big)$.
\end{lemma}

\begin{proof}
	We proceed by induction on $k$. Let $k=1$ and let $x\in\Fq$ be a solution of the equation $f(x)=b$. Then $f(x)^{\frac{q-1}{m}}=x^{\frac{q-1}{m}n} h(x^{\frac{q-1}{m}})^{\frac{q-1}{m}}=b^{\frac{q-1}{m}}$. Since $f(x)$ is $m$-nice, it follows that $x^{\frac{q-1}{m}}=\psi_f^{(-1)}\big(b^{\frac{q-1}{m}}\big)$. Suppose that the result follows for an integer $k\ge 1$ and let $x\in\Fq$ be a solution of the equation $f^{(k+1)}(x)=b$. By induction hypothesis, $f(x)^{\frac{q-1}{m}}=\psi_f^{(-k)}\big(b^{\frac{q-1}{m}}\big)$. Therefore, $x^{\frac{q-1}{m}n} h(x^{\frac{q-1}{m}})^{\frac{q-1}{m}}=\psi_f^{(-k)}\big(b^{\frac{q-1}{m}}\big)$, which implies that $x^{\frac{q-1}{m}}=\psi_f^{(-(k+1))}\big(b^{\frac{q-1}{m}}\big)$, since $f(x)$ is $m$-nice.
\end{proof}

\begin{proposition}\label{item609}
	Let $a\in\Fq^*$. If $f(x)$ is $m$-nice, then $R_a\left(\mathcal{G}\left(f / \Fq\right)\right)$ is $\gcd_{n}(\nu)$-regular.
\end{proposition}

\begin{proof}
	Let $b$ be a vertex of $R_a\left(\mathcal{G}\left(f / \Fq\right)\right)$ and $k$ be a positive integer. The number of $k$-distant predecessors of $b$ is equal to the number of solutions of the equation
	 \begin{equation}\label{item605}
	 	f^{(k)}(x)=b
	 \end{equation}
	  over $\Fq$. By Lemma~\ref{item604}, a solution $x\in\Fq$ of Equation~\eqref{item605} must satisfy the relation $x^{\frac{q-1}{m}}=\psi_f^{(-k)}\big(b^{\frac{q-1}{m}}\big)=:\xi\in\mu_m$. Let $\alpha$ be a primitive element of $\Fq$ and let $t$ be an integer such that $\xi=\alpha^{\frac{q-1}{m}t}$. Then the equality $x^{\frac{q-1}{m}}=\alpha^{\frac{q-1}{m}t}$ implies $x=\alpha^{t+m\ell}$ for some $l=1,\dots,\tfrac{q-1}{m}$. Now, Equations~\eqref{item606} and \eqref{item605} states that
	  \begin{equation}\label{item607}
	  	\alpha^{(t+m\ell)n^k}c=b,
	  \end{equation}
	  where $c=\prod_{i=0}^{k-1} h\left(\psi_f^{(i)}\left(\xi\right)\right)^{n^{k-i-1}}$ and $l=1,\dots,\tfrac{q-1}{m}$. In order to complete the proof, we will prove the following statement.
	   
	  \textbf{Claim.} The number of integers $\ell\in\{1,\dots,\tfrac{q-1}{m}\}$ satisfying Equation~\eqref{item607} is equal to either $0$ or $\gcd\big(n^k,\tfrac{q-1}{m}\big)$.
	  
	  \textit{Proof of the claim.} Let $u$ be an integer such that $b/c=\alpha^u$. We want to compute the number of integers $\ell=1,\dots,\tfrac{q-1}{m}$ such that $\alpha^{(t+m\ell)n^k}=\alpha^u$, that is
	  \begin{equation}\label{item608}
	  (t+m\ell)n^k\equiv u\pmod{q-1}.
	  \end{equation}
	  Assume that this equation has at least one solution. Then $\gcd(m n^k,q-1)$ must divide $u-tn^k$. In this case, Equation~\eqref{item608} becomes
	  $$\tfrac{n^k}{\gcd(n^k,s)} \ell\equiv \tfrac{u-tn^k}{\gcd(mn^k,q-1)}\pmod{\tfrac{q-1}{m\gcd(n^k,s)}},$$
	  where $s=\tfrac{q-1}{m}$. Now, since $\tfrac{n^k}{\gcd(n^k,s)}$ is relatively prime to $\tfrac{q-1}{m\gcd(n^k,s)}$, there exists exactly one solution $\ell$ to the above equation in the interval $\big[1,\tfrac{q-1}{m\gcd(n^k,s)}\big]$. Therefore, Equation~\eqref{item608} has $\gcd(n^k,s)$ solutions, which proves our claim.
	  
	  By the Claim, the number of $k$-distant predecessors of $b$ is either $0$ or $\gcd(n^k,s)$, which is the $k$-th entry of $\gcd_{n}(\nu)$. Since $b\in\Fq^*$ and $k$ were taken arbitrarily, the proof of our assertion is complete.
	  
\end{proof}

 We recall a classic result from Number Theory that will be used in the proof of Theorem~\ref{item601}.

\begin{theorem}\cite[M{\"o}bius inversion formula]{ireland1982classical}\label{item615} Let $G(u)=\sum_{d\mid u} g(d)$. Then $g(u)=\sum_{d\mid u}\mu(u/d)G(d)$.

\end{theorem}

Now we are able to prove the main results of the paper.

\section{Functional graph of polynomial maps}\label{item624}

In this section, we provide the proof of our main results. We start by proving Theorem~\ref{item603}.

\subsection{Proof of Theorem~\ref{item603}} Let $\{0,w_1,\dots,w_{r_1(q-1)/m}\}$ denote the set of children of $0$ in $\mathcal{G}\left(f / \Fq\right)$, that consists of the solutions of the the equation
$$x^n h(x^{\frac{q-1}{m}})=0$$
over $\Fq$. For each $j=1,\dots,\tfrac{q-1}{m} r_1$, let $T_j=R_{w_j}\left(\mathcal{G}\left(f / \Fq\right)\right)$. Since the image $f(0)$ is equal to $0$, the vertex $0$ is the single vertex of the cyclic part of this component and then each $T_j$ is a tree. Therefore, $\mathcal{G}^{(0)}_{f / \Fq}=\cyc(1,T)$, where
\begin{equation}\label{item612}
	T=\left\langle\bigoplus_{j=1}^{r_1(q-1)/m} T_j\right\rangle.
\end{equation}
By Proposition~\ref{item609} and Lemma~\ref{item610}, each $T_j$ is isomorphic to $\Ta_{\gcd_{n}(\nu)}^{i_j}$, where $i_j$ is the depth of $T_j$. Therefore, we only need to determine the cardinality of each set
$$A_i=\big\{j\in\{1,\dots,r_1(q-1)/m\}:T_j\text{ is isomorphic to }\Ta_{\gcd_{n}(\nu)}^{i}\big\}.$$
In order to do that, we define the set
$$B_i=\{j\in\{1,\dots,r_1(q-1)/m\}:T_j\text{ has a vertex with depth }i\}.$$
We observe that $|A_i|=|B_i|-|B_{i+1}|$. Let $z_1,\dots,z_{r_1}$ be the elements in $\mu_m$ that are solutions of the equation $h(x)=0$. By Lemma~\ref{item604}, any vertex $x$ of $T_j$ with depth $i$ is a solution of the equation
$$x^{\frac{q-1}{m}}=\psi_f^{(-i)}\big(w_j^{\frac{q-1}{m}}\big).$$
On the other hand, since $f(x)$ is $m$-nice, any solution of the above equation must be a vertex with depth $i$ of $T_j$ for some $j=1,\dots,r_1(q-1)/m$. Therefore, we are interested in the number of solutions of the equations
\begin{equation}\label{item611}
	x^{\frac{q-1}{m}}=\psi_f^{(-i)}\big(z_\ell\big),
\end{equation}
 where $\ell=1,\dots,r_1$. Taking $x^{\frac{q-1}{m}}=\xi\in\mu_m$, Equation~\eqref{item611} becomes $$\xi=\psi_f^{(-i)}\big(z_\ell\big),$$ that has a solution (for some $\ell$) for $r_{i+1}$ distinct values $\xi\in\mu_m$. Therefore, the number of solutions of Equation~\eqref{item611} is equal to $r_{i+1}\times\tfrac{q-1}{m}$. Since $T_j$ is $\gcd_{n}(\nu)$-regular, the number of $i$-distant predecessors of $w_j$ in $T_j$ equals either $0$ or $d_i=\gcd(n^i,\nu)$. Therefore,
 $$|B_i|=\frac{r_{i+1}(q-1)}{m d_i}.$$
Now, it follows from Equation~\eqref{item612} that
$$T=\left\langle\bigoplus_{i=0}^{\infty} \left(\tfrac{(q-1)r_{i+1}}{md_{i}}-\tfrac{(q-1)r_{i+2}}{md_{i+1}}\right)\times\Ta_{\gcd_{n}(\nu)}^{i}\right\rangle.$$
Since there exist at most $m$ elements in $\mu_m$, the depth of sum of these tree is at most $m-1$, and therefore we may assume without loss of generality that $i\le m-1$, which completes the proof of our assertion. $\hfill\qed$

We are now able to prove the main result of the paper.

\subsection{Proof of Theorem~\ref{item601}}

We recall that each connected component of $\mathcal{G}^{(1)}_{f / \Fq}$ is composed by a cycle and each vertex of this cycle is a non-null element of $\Fq$ that is the root of a tree.  By Lemma~\ref{item613} and Proposition~\ref{item609}, any of such trees is isomorphic to $\Ta_{\gcd_{n}(\nu)}$. Therefore, it only remains to determine what are the cycles in $\mathcal{G}^{(1)}_{f / \Fq}$. Our goal now is to determine how many cycles there exist with length $\ell$.

By Lemma~\ref{item604}, we have that the length of a cycle is closely related to the dynamics of $\psi_f$ over $\mu_m$. Indeed, if $f^{(\ell)}(a)=a$ for a positive integer $\ell$, then $a^{\frac{q-1}{m}}=\psi_f^{(-\ell)}(a^{\frac{q-1}{m}})$, which implies that $\psi_f^{(\ell)}(a^{\frac{q-1}{m}})=a^{\frac{q-1}{m}}$, since $f$ is $m$-nice. In this case, if $a^{\frac{q-1}{m}}\in\ S_i$, then $k_i\mid \ell$. Furthermore, any vertex $b$ in the same cycle of $a$ satisfies $b^{\frac{q-1}{m}}\in S_i$. In particular, that means that the cycles whose dynamics are related to two different sets $S_i$ and $S_j$ are not connected. Therefore, we may determine each one of this cycles separately.  For a positive integer $u$ and a fixed $i\in\{1,\dots,t\}$, let 
$$A_i(d)=\{a\in\Fq:a^{\frac{q-1}{m}}=\xi_i, f^{(d k_i)}(a)=a\}$$
and
$$B_i(u)=\{a\in\Fq:a^{\frac{q-1}{m}}=\xi_i, u \text{ is the least positive integer such that } f^{(u k_i)}(a)=a\}.$$
In order to determine how many cycles (with vertices $a$ such that $a^{\frac{q-1}{m}}=\xi_i)$) there exist with length $u k_i$, we need to determine $|B_i(u)|$. We note that an element $a\in A_i(u)$ is a vertex in a cycle whose length $s$ divides $u$, then $|A_i(u)|=\sum_{d\mid u} |B_i(d)|$. The M{\"o}bius inversion formula (Theorem~\ref{item615}) implies that
\begin{equation}\label{item617}
|B_i(u)|=\sum_{d\mid u} \mu(u/d)|A_i(d)|.
\end{equation}
 We now compute the value $|A_i(d)|$. In order to do so, let $a\in A_i(d)$. Since $a^{\frac{q-1}{m}}=\xi_i$ and $f^{(d k_i)}(a)=a$, it follows that $a=\alpha^{sm+r_i}$ for some integer $s\in\{1,\dots,\tfrac{q-1}{m}\}$ and then Equation~\eqref{item606} states that
$$\big(\alpha^{sm+r_i}\big)^{n^{dk_i}} \prod_{j=0}^{d k_i-1} h\left(\psi_f^{(j)}\left(\xi\right)\right)^{n^{dk_i-j-1}}=\alpha^{sm+r_i}.$$
Since $\alpha^{\ell_i}=\prod_{j=0}^{k_i-1}h\big(\psi_f^{(j)}(\xi_i)\big)^{n^{k_i-j-1}}$, the previous equations becomes
$$	\big(\alpha^{sm+r_i}\big)^{n^{dk_i}}\alpha^{\ell_i (1+n^{k_i}+\dots+n^{(d-1)k_i})}=\alpha^{sm+r_i}.$$
Looking at the exponents in this equation and doing some algebraic manipulations, it follows that
\begin{equation}\label{item619}
	\frac{n^{d k_i}-1}{n^{k_i}-1}\left((sm+r_i)(n^{k_i}-1)+\ell_i\right)\equiv 0\pmod{q-1}.
\end{equation}
 By using the same arguments used along the proof of Proposition~\ref{item609}, one can prove that that number of solutions $s\in\{1,\dots,\tfrac{q-1}{m}\}$ of the previous equations equals
 {\small $$\tau_i(d):=\begin{cases}
		\gcd\big(\tfrac{q-1}{m},n^{dk_i}-1\big),&\text{ if }\gcd\big(q-1,(n^{dk_i}-1)m\big)\mid \big(l_i\big(\tfrac{n^{dk_i}-1}{n^{k_i}-1}\big)+r_i(n^{dk_i}-1)\big);\\
		0,&\text{ otherwise.}\\
	\end{cases}$$ }
On the other hand, each solution $s\in\{1,\dots,\tfrac{q-1}{m}\}$ of Equation~\eqref{item619} yields an element in $A_i(d)$ and, therefore,
\begin{equation}\label{item614}
	|A_i(d)|=\tau_i(d).
\end{equation}
By Equations~\eqref{item617} and~\eqref{item614}, it follows that
$$|B_i(u)|=\sum_{d\mid u} \mu(u/d)\tau_i(d).$$
Now it only remains to prove if $u$ is an integer for which there exist a cycle in $\mathcal{G}^{(1)}_{f / \Fq}$ with length $uk_i$, then $u\mid \ord_{\scalebox{0.7}{$w'(n^{k_i}-1)$}}(n^{\scalebox{0.7}{$k_i$}})$. In order to do so, we observe that if $a=\alpha^{sm+r_i}$ is an element in a cycle of length $uk_i$, then Equation~\eqref{item619} implies that $u$ is the least integer such that 
 $$\frac{n^{u k_i}-1}{n^{k_i}-1}\left((sm+r_i)(n^{k_i}-1)+\ell_i\right)\equiv 0\pmod{q-1},$$
 which implies that $u= \ord_{\scalebox{0.7}{$d(n^{k_i}-1)$}}(n^{\scalebox{0.7}{$k_i$}})$, where the is a divisor of $q-1$ coprime to $n$. In particular, $d\mid w'$ so that $\ord_{\scalebox{0.7}{$d(n^{k_i}-1)$}}(n^{\scalebox{0.7}{$k_i$}})\mid \ord_{\scalebox{0.7}{$w'(n^{k_i}-1)$}}(n^{\scalebox{0.7}{$k_i$}})$, which completes the proof of our theorem. $\hfill\qed$

\printbibliography
 
\end{document}